\documentclass[12pt]{article} 
\usepackage{amsmath, graphicx}
\usepackage{amstext,amsfonts,amsbsy,eucal,amssymb, amsthm,color}

\usepackage{pgf,tikz,pgfplots}
\usetikzlibrary{arrows}

\definecolor{light-gray}{gray}{0.80}
\definecolor{light-gray2}{gray}{0.70}

\advance \topmargin by -\headheight
\advance \topmargin by -\headsep
     
\evensidemargin \oddsidemargin
\newtheorem{theorem}{Theorem}[section]
\newtheorem{proposition}[theorem]{Proposition}

\theoremstyle{definition}
\newtheorem{definition}[theorem]{Definition}

\newtheorem{remark}[theorem]{Remark}

\def\Z{\operatorname{\mathbb{Z}}}

\def\C{\operatorname{\mathbb{C}}}
\def\P{\operatorname{\mathbb{P}}}

\def\cX{\operatorname{\mathcal{X}}}
\def\cY{\operatorname{\mathcal{Y}}}

\def\cH{\operatorname{\mathcal{H}}}
\def\cE{\operatorname{\mathcal{E}}}

\def\ve{\varepsilon}

\begin{document}
\begin{center}
{\large {\bf Space of initial conditions for the four-dimensional Fuji-Suzuki-Tsuda system}}
\end{center}

\begin{center}

{\bf Tomoyuki Takenawa}\\
\medskip
{\small
Faculty of Marine Technology, Tokyo University of Marine Science and Technology\\ 2-1-6 Etchu-jima, Koto-ku, Tokyo, 135-8533, Japan\\
E-mail: takenawa@kaiyodai.ac.jp}
\end{center}

\medskip

\begin{abstract}
A geometric study for an integrable 4-dimensional dynamical system so called the Fuji-Suzuki-Tsuda system is given.
By the resolution of indeterminacy, the group of its B\"aklund transformations is lifted to a group of pseudo-isomorphisms between rational varieties obtained from $(\P^1)^4$ by blowing-up along eight 2-dimensional subvarieties and four 1-dimensional subvarieties. The root basis is realised in the N\'eron-Severi bilattices.
A discrete Painlev\'e system with quadratic degree growth is also realised as its translational element. 
\end{abstract}

\section{Introduction}
\subsection{Background and the results}
The Painlev\'e equations are nonlinear second-order ordinary differential equations whose solutions are meromorphic except some fixed points, but not reduced to known functions such as solutions of linear ordinary differential equations or Abel functions.
In \cite{Okamoto1977} Okamoto introduced the notion of space of initial conditions
where the flow of each Painlev\'e equation is regularised on a family of rational algebraic surfaces (minus some subvarieties called vertical leaves) even around poles.
In \cite{Sakai2001} Sakai extended this notion to the discrete case and used it to obtain symmetry group of the equations. A benefit of this approach is that the root systems can be realised geometrically in the Picard group on the surfaces.

In recent years, research on four-dimensional Painlev\'e systems has been progressed mainly from the viewpoint of isomonodromic deformation of linear equations \cite{Sakai2018} and pointed out that there are 4 master equations in the sense that other equations can be obtained from them by limiting procedure \cite{KNS2018}. 
The four-dimensional Fuji-Suzuki-Tsuda system is one of these 4 master equations. 
In \cite{FS2010, Suzuki2013}  Suzuki and Fuji obtained the $2N$-dimensional ($N=1,2,3,\cdots$) system by a reduction from so called the Drinfeld-Sokolov hierarchies of type $A$ and in \cite{Tsuda2014} Tsuda obtained it from so called the UC-hierarchies. 

In this paper, starting from known B\"aclund transformations, we construct the space of initial conditions for the 4D Fuji-Suzuki-Tsuada system.
By resolution of indeterminacy, the B\"aklund transformations are lifted to pseudo-isomorphisms between rational varieties obtained from $(\P^1)^4$ by blowing-up along eight 2-dimensional subvarieties and four 1-dimensional subvarieties\footnote{Sasano has also constructed the space of initial conditions in \cite{Sasano2007}, where he started from a four-dimensional analog of the Hirzebruch surface. In our study, the structure of the Picard group becomes simple, since we start from the direct product of $\P^1$ instead of the analog of the  Hirzebruch surface.}. The root basis is realised in the N\'eron-Severi bilattice. A discrete Painlev\'e system with quadratic degree growth is also realised as its translational element.

\subsection{Fuji-Suzuki-Tsuda system and its B\"acklund transformations}
The 4D Fuji-Suzuki-Tsuada system is a Hamiltonian system 
\begin{align}
& \frac{dq_i}{dt}=\frac{\partial H}{\partial p_i},\quad
\frac{dp_i}{dt}=-\frac{\partial H}{\partial q_i},\quad
(i=1,2)
\end{align}
with
\begin{align}
t(t -1)H = &H_{{\rm VI}}(q_1, p_1; a_2, a_0 + a_4, a_3 + a_5 - \eta, \eta a_1)\nonumber \\ \nonumber
&+ H_{\rm VI}(q_2, p_2; a_0 + a_2, a_4, a_1 + a_3 - \eta, \eta a_5)\\
&+ (q_1 - t)(q_2 - 1) \{(q_1p_1 + a_1)p_2 + p_1(p_2q_2 + a_5)\}
\end{align}
and $a_0+a_1+\cdots+a_5=1$,
where $H_{\rm VI}$ is the polynomial Hamiltonian of  the sixth Painlev\'e equation introduced by Okamoto in \cite{Okamoto1980}:\footnote{
Precisely saying, $H_{\rm VI}$ is the Hamiltonian introduced by Okamoto multiplied by $t(t-1)$.}
\begin{align}
H_{\rm VI}(q, p; a, b, c, d) =& q(q - 1)(q - t)p^2 - \{(a - 1)q(q - 1)\nonumber \\
&+ bq(q - t) + c(q - 1)(q - t)\} p + dq.
\end{align}
This equation has B\"acklund transformations, i.e. transformations of variables
which keep the equation except parameters as the following two tables\footnote{
~ $s_i$'s are reported in \cite{FS2010}, while
$\pi$ is in \cite{Suzuki2013}, where a misprint has been corrected, and 
$\rho$ is in \cite{Tsuda2014}.}.

\begin{center}
Actions on parameters
\begin{tabular}{|c|c|c|c|c|c|c|c|c|}\hline
\rule{0pt}{2.5ex} & $\bar{a}_0$&$\bar{a}_1$&$\bar{a}_2$&$\bar{a}_3$&$\bar{a}_4$&$\bar{a}_5$&$\bar{\eta}$&$\bar{t}$\\\hline 
$s_{0}$&$-a_0$&$a_0+a_1$&$a_2$&$a_3$&$a_4$&$a_0+a_5$&$\eta+a_0$&$t$\\ \hline
$s_{1}$&$a_0+a_1$&$-a_1$&$a_1+a_2$&$a_3$&$a_4$&$a_5$&$\eta-a_1$&$t$\\ \hline
$s_{2}$&$a_0$&$a_1+a_2$&$-a_2$&$a_2+a_3$&$a_4$&$a_5$&$\eta+a_2 $&$t$\\ \hline
$s_{3}$&$a_0$&$a_1$&$a_2+a_3$&$-a_3$&$a_3+a_4$&$a_5$&$\eta-a_3$&$t$\\ \hline
$s_{4}$&$a_0$&$a_1$&$a_2$&$a_3+a_4$&$-a_4$&$a_4+a_5$&$\eta+a_4$&$t$\\ \hline
$s_{5}$&$a_0+a_5$&$a_1$&$a_2$&$a_3$&$a_4+a_5$&$-a_5$&$\eta-a_5$&$t$\\ \hline
$\pi$ 
&$a_1$&$a_2$&$a_3$&$a_4$&$a_5$&$a_0$&$-\eta$&$t^{-1}$\\ \hline
$\rho$ &$a_0$&$a_5$&$a_4$&$a_3$&$a_2$&$a_1$&$\eta$&$t^{-1}$\\ \hline
\end{tabular}
\end{center}
\clearpage
\begin{center}
Actions on dependent variables
\begin{tabular}{|c|c|c|c|c|}\hline
\rule{0pt}{2.5ex} & $\bar{q}_1$&$\bar{q}_2$&$\bar{p}_1$&$\bar{p}_2$\\\hline 
$\rule{0pt}{2.5ex} s_{0}$&$q_1$&$q_2$&$p_1-a_0(q_1-q_2)^{-1}$&$p_2+a_0(q_1-q_2)^{-1}$\\\hline
$\rule{0pt}{2.5ex} s_{1}$&$q_1+a_1p_1^{-1}$&$q_2$&$p_1$&$p_2$\\\hline
$\rule{0pt}{2.5ex} s_{2}$&$q_1$&$q_2$&$p_1-a_2(q_1-t)^{-1}$&$p_2$\\\hline
$\rule{0pt}{2.5ex} s_{3}$&$q_1+a_3q_1B^{-1}$&$q_2+a_3q_2B^{-1}$&
$p_1-a_3p_1A^{-1}$&$p_2-a_3p_2A^{-1}$\\\hline
$\rule{0pt}{2.5ex} s_{4}$&$q_1$&$q_2$&$p_1$&$p_2-a_4(q_2-1)^{-1}$\\\hline
$\rule{0pt}{2.5ex} s_{5}$&$q_1$&$q_2+a_5 p_2^{-1}$&$p_1$&$p_2$\\\hline
$\rule{0pt}{2.5ex} \pi $&$C_1C_t^{-1} $&$C_2C_t^{-1} $&$-(q_1-t)C_t(t-1)^{-1} $&$-(q_2-q_1)C_t(t-1)^{-1}$\\\hline
$\rule{0pt}{2.5ex} \rho$&$q_2 t^{-1}$&$q_1 t^{-1}$&$p_2 t$&$p_1 t$\\\hline
\end{tabular}
\end{center}
where $A=q_1p_1+q_2p_2+\eta$, $B=A-a_3$,  $C_1=A-p_1-p_2$,
$C_2=A-tp_1-p_2$ and $C_t=A-tp_1-tp_2$.

These transformations constitute so called the extended affine Weyl group of type $A_5^{(1)}$, whose fundamental relations are 
\begin{align*}
&s_i^2 =\rho^2=\pi^6={\rm identity}\\
 &s_i \circ s_{i+1} \circ s_i = s_{i+1} \circ s_i \circ s_{i+1},\qquad s_i \circ s_j = s_j\circ s_i \quad \mbox(|i-j|>1)\\
&s_{i} \circ \pi= \pi \circ s_{i+1},\qquad
s_i \circ \rho =\rho \circ s_{6-i}, \qquad (\pi \circ \rho)^2={\rm identity},
\end{align*}
where indices are considered to be cyclic as $\alpha_{i+6}=\alpha_i$. 

In order to construct its space of initial conditions without blowing-downs, we introduce a new coordinate system $(q_i,r_i)=(q_i, q_i p_i)$ ($i=1,2$), exactly the same manner with the case of two-dimensional Painlev\'e equations.

The action on dependent variables become

\begin{center}
Actions on new dependent variables
\begin{tabular}{|c|c|c|c|c|}\hline
\rule{0pt}{2.5ex} & $\bar{q}_1$&$\bar{q}_2$&$\bar{r}_1$&$\bar{r}_2$\\\hline 
$\rule{0pt}{2.5ex} s_{0}$&$q_1$&$q_2$&$r_1-a_0q_1(q_1-q_2)^{-1}$&$r_2+a_0q_2(q_1-q_2)^{-1}$\\\hline
$\rule{0pt}{2.5ex} s_{1}$&$q_1(1+a_1 r_1^{-1})$&$q_2$&$r_1+a_1$&$r_2$\\\hline
$\rule{0pt}{2.5ex} s_{2}$&$q_1$&$q_2$&$r_1-a_2q_1(q_1-t)^{-1}$&$r_2$\\\hline
$\rule{0pt}{2.5ex} s_{3}$&$q_1+a_3q_1D^{-1}$&$q_2+a_3q_2D^{-1}$&
$r_1$&$r_2$\\\hline
$\rule{0pt}{2.5ex} s_{4}$&$q_1$&$q_2$&$r_1$&$r_2-a_4q_2(q_2-1)^{-1}$\\\hline
$\rule{0pt}{2.5ex} s_{5}$&$q_1$&$q_2(1+a_5 r_2^{-1})$&$r_1$&$r_2+a_5$\\\hline
$\rule{0pt}{2.5ex} \pi $&$E_1E_t^{-1} $&$E_2E_t^{-1} $&$-(q_1-t)E_1F^{-1} $&$-(q_2-q_1)E_2 F^{-1}$\\\hline
$\rule{0pt}{2.5ex} \rho$&$q_2 t^{-1}$&$q_1 t^{-1}$&$r_2$&$r_1$\\\hline
\end{tabular}
\end{center}
where $D=r_1+r_2-a_3+\eta$, $E_1=q_1q_2(r_1+r_2+\eta)-q_2r_1-q_1r_2$, 
$E_2=q_1q_2(r_1+r_2+\eta)-tq_2r_1-q_1r_2$, $E_t=q_1q_2(r_1+r_2+\eta)-tq_2r_1-tq_1r_2$ and $F=q_1q_2(t-1)$.

\subsection{Basic facts}
In this paper, we use the following basic facts; see \S~2 of \cite{CT2019} for details. 
 
Let $\cX$ and $\cY$ be smooth projective varieties.  For a birational map $f:\cX \to \cY$, let $I(f)$ denote the indeterminate set (i.e. the set of points where $f$ is not defined) of $f$ in $\cX$.
 
We say a sequence of birational maps $\varphi_n: \cX_n \to\cX_{n+1}$ for smooth projective varieties $\cX_n$ ($n\in \Z$)   
to be algebraically stable if $$\left(\varphi_{n+k-1}  \circ \cdots \circ  \varphi_{n+1}  \circ \varphi_n\right)^* =  \varphi_n^* \circ \varphi_{n+1}^* \circ \cdots \circ \varphi_{n+k-1}^* $$ 
holds as a mapping from the Picard group of $\cX_{n+k}$ to that of $\cX_n$
for any integers $n$ and $k\geq 1$.

\begin{proposition}[\cite{BK2008, Bayraktar2012}] \label{AS2}
A sequence of birational maps $\varphi_n: \cX_n \to\cX_{n+1}$ for smooth projective varieties $\cX_n$ ($n\in {\mathbb Z}$)   
is algebraically stable if and only if
there do not exist integers $n$ and $k\geq 1$ and a divisor $D$ on $\cX_{n-1}$ such that $\varphi(D\setminus I(\varphi_{n-1} ))\subset I(\varphi_{n+k-1}  \circ \cdots \circ  \varphi_{n+1}  \circ \varphi_n)$.\footnote{
This statement is a non-autonomous analog of a proposition shown in \cite{BK2008, Bayraktar2012}. The proof does not change except in notations.}
\end{proposition}

We call a birational mapping $f:\cX\to \cY$ a {\it pseudo-isomorphism}
if $f$ is isomorphic except on finite number of subvarieties of codimension two at least. This condition is equivalent to that there is no prime divisor pulled back to the zero divisor by $f$ or $f^{-1}$. Hence, if $\varphi_n$ is a pseudo-isomorphism for each $n$, then $\{\varphi_n\}_{n\in \Z}$ and $\{\varphi_n^{-1}\}_{n\in \Z}$ are algebraically stable.  

\begin{proposition}[\cite{DO1988}]\label{NSauto}
Let $\cX$ and $\cY$ be smooth projective varieties and
$\varphi$ a pseudo-isomorphism from $\cX$ to $\cY$. Then
$\varphi$ acts on the N\'eron-Severi bi-lattice as an automorphism preserving the intersections.
\end{proposition}
The N\'eron-Severi bi-lattice of a smooth rational variety $\cX$ is isomorphic to $H^2(\cX, \Z) \times H_2(\cX, \Z)$ which is explicitly given in the following.\\

\noindent {\it Blowup of a direct product of $\P^1$}\\
In accordance with \cite{TT2009}, we take the basis of the N\'eron-Severi bi-lattice as follows. 

Let $\cX$ be a rational variety obtained by $K$ successive blowups from $(\P^1)^N$  and 
$$({\bf x}_1,{\bf x_2},\dots, {\bf x}_N)=(x_{10}:x_{11},~ x_{20}:x_{21},~\cdots,~ x_{N0}:x_{N1})$$
the direct product of homogeneous coordinate chart. 
Let $\cH_i$ denote the total transform of the class of a hyper-plane ${\bf c}_i{\bf x}_i=c_{i0}x_{i0}+c_{i1}x_{i1}=0$, where ${\bf c}_i=(c_{i0}:c_{i1})$ is a constant vector in $\P^{1}$, and $\cE_k$ the total transform of the $k$-th exceptional divisor class. Let $h_i$ denote the total transforms of the class of a line
\begin{align*}
\{{\bf x}~|~&\forall j\neq i, (x_{j0}:x_{j1})=(c_{j0}:c_{j1})\},
\end{align*}
where 
 ${\bf c}_j=(c_{j0}:c_{j1})$'s ($j \neq i$) are constant vectors in $\P^1$, and $e_k$ the class of a line in a fiber of the $k$-th blow-up. Note that the exceptional divisor for a blowing-up along a $d$-dimensional subvariety $V$ is isomorphic to $V\times \P^{N-d-1}$, where $\P^{N-d-1}$ is a fiber.   

Then the Picard group $\simeq H^2(\mathcal{X},\Z)$ and its Poincar\'e dual $\simeq H_2(\mathcal{X},\Z)$ are lattices
\begin{align}\label{NSbasis}
H^2(\mathcal{X},\Z)=\bigoplus_{i=1}^n \Z \cH_i \oplus \bigoplus_{k=1}^{K} \Z \cE_k,\quad 
H_2(\mathcal{X},\Z)=\bigoplus_{i=1}^n \Z h_i \oplus \bigoplus_{k=1}^{K} \Z e_k
\end{align}
and the intersection form is given by
\begin{align}
\langle \cH_i, h_j\rangle = \delta_{ij},\quad 
\langle \cE_k, e_l\rangle = -\delta_{kl},\quad
\langle \cH_i, e_k\rangle =0,\quad
\langle \cE_k, h_i\rangle =0.
\end{align} 

\noindent {\it Degree of a mapping}\\
Let $\psi$ be a rational mapping from $\C^N$ to itself: 
$$ \psi: (\bar{x}_1,\dots,\bar{x}_N)= (\psi_1(x_1,\cdots,x_N),\dots,\psi_N(x_1,\cdots,x_N)).$$ 
The degree of $\bar{x_i}$ of $\psi$ with respect $x_j$ is defined as
the degree of $\psi_i$ as a rational function of $x_j$, i.e. the maximum of degrees of
numerator and denominator.    
If $\C^N$ is compactified as $(\P^1)^N$, 
the degree of $\bar{x_i}$ of $\psi$ with respect $x_j$ is given by
the coefficient of $\cH_j$ in $\psi^*(\cH_i)$. This formula also holds when
$(\P^1)^N$ is blown-up if $\cH_i$ denotes the total transform with respect to blowing-up as the above settings.

\section{Construction of the space of initial conditions}

First of all, let us compactify the phase space $\{(q_1,q_2,r_1,r_2)\in \C^4\}$
to $(\P^1)^4$ by introducing new coordinates $Q_1=q_1^{-1}$,  $Q_2=q_2^{-1}$, 
 $R_1=r_1^{-1}$,  $R_2=r_2^{-1}$ around $q_1=\infty$, $q_2=\infty$ and so on. 

Next, we search  $(\P^1)^4$ for hyper-surfaces which are contracted to lower dimensional subvarieties. Such hyper-surfaces appear as a factor of the numerator of Jacobian of a map $\partial(\bar{q}_1,\bar{q}_2, \bar{r}_1,\bar{r}_2)/\partial(q_1,q_2,r_1,r_2)$. However, note that $(\P^1)^4$ has essentially $2^4=16$ charts, and the Jacobian should be considered between all the pairs of charts. 

For example, the Jacobian  $\partial(\bar{q}_1,\bar{q}_2, \bar{r}_1,\bar{r}_2)/\partial(q_1,q_2,r_1,r_2)$ of $s_0$ is a constant of 1 in the original chart, while it becomes nontrivial on another chart as
$$\partial(\bar{q}_1,\bar{q}_2, \bar{R}_1,\bar{R}_2)/\partial(q_1,q_2,r_1,r_2)
=\frac{(q_1 - q_2)^4}{(a_0 q_1 - q_1 r_1 + q_2 r_1)^2 (a_0 q_2 + q_1 r_2 - q_2 r_2)^2}.$$
Thus, the image of the generic part of hyper-surface $q_1-q_2=0$ is contracted to some lower dimensional subvariety.
Indeed, substituting $q_2=q_1+\ve$ to $(\bar{q}_1,\bar{q}_2,\bar{R}_1,\bar{R}_2)$, we have
$$(\bar{q}_1,\bar{q}_2,\bar{R}_1,\bar{R}_2)=\left(q_1,q_1+\ve, \frac{\ve}{a_0q_1}+O(\ve^2),-\frac{\ve}{a_0q_1}+O(\ve^2)\right),$$
where $O$ is the big O asymptotic notation, 
Hence, we can see that the generic part of hyper-surface $q_1-q_2=0$ is contracted to a two-dimsensional variety
\begin{align}\label{image of s0}
&\{(\bar{q}_1,\bar{q}_2,\bar{R}_1,\bar{R}_2) \in \C^4~|~
\bar{q}_1-\bar{q}_2=0,\  \bar{R}_1= \bar{R}_2=0\}.
\end{align}

Similarly, for $s_1$, from
$$\partial(\bar{q}_1,\bar{q}_2, \bar{r}_1,\bar{r}_2)/\partial(q_1,q_2,r_1,r_2)
=\frac{a_1 + r_1}{r_1}$$ 
and
$$\partial(\bar{Q}_1,\bar{q}_2, \bar{r}_1,\bar{r}_2)/\partial(q_1,q_2,r_1,r_2)
=-\frac{r_1}{q_1^2 (a_1 + r_1)},$$ 
we see that $r_1+a_1=0$ is contracted to
\begin{align}
&\{(\bar{q}_1,\bar{q}_2,\bar{r}_1,\bar{r}_2) \in \C^4~|~
\bar{q}_1=0,  \bar{r}_1=0\}
\end{align} 
and $r_1=0$ is contracted to
\begin{align}
&\{(\bar{Q}_1,\bar{q}_2,\bar{r}_1,\bar{r}_2) \in \C^4~|~
\bar{Q}_1=0,  \bar{r}_1+\bar{a}_1=0\}.
\end{align} 

For $s_2$, from
$$\partial(\bar{q}_1,\bar{q}_2, \bar{R}_1,\bar{r}_2)/\partial(q_1,q_2,r_1,r_2)
=-\frac{(q_1-t)^2}{(a_2 q_1 - q_1 r_1 + r_1 t)^2},$$ 
we see $q_1-t=0$ is contracted to
\begin{align}
&\{(\bar{q}_1,\bar{q}_2,\bar{R}_1,\bar{r}_2) \in \C^4~|~
\bar{q}_1-\bar{t}=0,  \bar{R}_1=0\}.
\end{align}

For $s_3$, from
$$\partial(\bar{q}_1,\bar{q}_2, \bar{r}_1,\bar{r}_2)/\partial(q_1,q_2,r_1,r_2)
=\frac{(r_1 + r_2+\eta)^2}{(r_1 + r_2 -a_3 + \eta)^2}$$
and 
$$\partial(\bar{Q}_1,\bar{Q}_2, \bar{r}_1,\bar{r}_2)/\partial(q_1,q_2,r_1,r_2)
=\frac{(r_1 + r_2 -a_3 + \eta)^2}{q_1^2 q_2^2  (r_1 + r_2+\eta)^2},$$ 
we see that
$r_1 + r_2 + \eta=0$ is contracted to
\begin{align}
&\{(\bar{q}_1,\bar{q}_2,\bar{r}_1,\bar{r}_2) \in \C^4~|~
\bar{q}_1=\bar{q}_2=0,  \bar{r}_1+\bar{r}_2-\bar{a}_3+\bar{\eta}=0\}
\end{align} 
and
$r_1 + r_2 -a_3 + \eta=0$ is contracted to
\begin{align}
&\{(\bar{Q}_1,\bar{Q}_2,\bar{r}_1,\bar{r}_2) \in \C^4~|~
\bar{Q}_1=\bar{Q}_2=0,  \bar{r}_1+\bar{r}_2+\bar{\eta}=0\}.
\end{align} 

For $s_4$, from
$$\partial(\bar{q}_1,\bar{q}_2, \bar{r}_1,\bar{R}_2)/\partial(q_1,q_2,r_1,r_2)
=-\frac{(q_2-1)^2}{(a_4 q_2 + r_2 - q_2 r_2)^2},$$ 
we see that $q_2-1=0$ is contracted to
\begin{align}
&\{(\bar{q}_1,\bar{q}_2,\bar{r}_1,\bar{R}_2) \in \C^4~|~
\bar{q}_2-1=0,  \bar{R}_2=0\}.
\end{align}

For $s_5$, from
$$\partial(\bar{q}_1,\bar{q}_2, \bar{r}_1,\bar{r}_2)/\partial(q_1,q_2,r_1,r_2)
=\frac{a_5+r_2}{r_2}$$
and 
$$\partial(\bar{q}_1,\bar{Q}_2, \bar{r}_1,\bar{r}_2)/\partial(q_1,q_2,r_1,r_2)
=-\frac{r_2}{q_2^2 (r_2+a_5)},$$ 
we see that $r_2+a_5=0$ is contracted to
\begin{align}
&\{(\bar{q}_1,\bar{q}_2,\bar{r}_1,\bar{r}_2) \in \C^4~|~
\bar{q}_2=0,  \bar{r}_2=0\}
\end{align} 
and $r_2=0$ is contracted to
\begin{align}
&\{(\bar{q}_1,\bar{Q}_2,\bar{r}_1,\bar{r}_2) \in \C^4~|~
\bar{Q}_2=0,  \bar{r}_2+\bar{a}_5=0\}.
\end{align} 

Let us take a closer look at the image \eqref{image of s0} of $q_1-q_2=0$ by $s_0$ 
and introduce new coordinates (by blowing-up)
\begin{align}\label{U11}
(w,q_2,u,v)=\left(\frac{q_1-q_2}{R_1}, q_2, R_1,\frac{R_2}{R_1}\right).
\end{align}
Then, from  
$$\partial(\bar{w},\bar{q}_2, \bar{u},\bar{v})/\partial(q_1,q_2,r_1,r_2)
=\frac{(q_1-q_2)^2}{(a_0q_2+q_1 r_2 -q_2 r_2)^2},$$
we see that $q_1-q_2=0$ is still contracted to 
\begin{align}
&\{ (\bar{w},\bar{q}_2,\bar{u},\bar{v}) \in \C^4~|~
\bar{w}-\bar{a}_0\bar{q}_2=0,  \bar{u}=0, \bar{v}+1=0\}.
\end{align} 
Let us introduce further new coordinates (by blowing-up)
\begin{align}\label{U12}
(w',q_2,u',v')=\left(\frac{w-a_0q_2}{u}, q_2, u,\frac{v+1}{u}\right).
\end{align}
Then, from  
$$\partial(\bar{w}',\bar{q}_2, \bar{u}',\bar{v}')/\partial(q_1,q_2,r_1,r_2)
=-1,$$
we see that there is no hyper-surface in the affine space of $(q_1,q_2,r_1,r_2)$ chart contracted to a subvariety whose generic part is included in the affine space of $(\bar{w}',\bar{q}_2,\bar{u}',\bar{v}')$ chart. (Charts \eqref{U11} and \eqref{U12} are the same with $U_{11}$ and $U_{12}$ below.)

By investigating other transformations in the same way, we can see that,
resolution of singularities requires two infinitesimally near blowing-ups (i.e. the center of the second blowing-up is included in the exceptional divisor of the first blowing up)
for $s_0$, $s_2$ or $s_4$, while it also requires two blowing-ups but not infinitesimally near for $s_1$, $s_3$ or $s_5$.

These blowing-ups are given by the following list.  
\begin{align}
\begin{array}{ll}
C_1: q_1=r_1=0  &U_1: (u_1, q_2,v_1, r_2)=(q_1, q_2 ,r_1q_1^{-1},r_2)\\
C_2: Q_1=r_1+a_1=0  &U_2: (u_2, q_2,v_2, r_2)=(Q_1, q_2 ,(r_1+a_1)Q_1^{-1},r_2)\\
C_3: q_2=r_2=0  &U_3: (q_1, u_3,r_1, v_3)=(q_1, q_2 ,r_1,r_2q_2^{-1})\\
C_4: Q_2=r_2+a_5=0  &U_4: (q_1, u_4,r_1, v_4)=(q_1, Q_2 ,r_1,(r_2+a_5)Q_2^{-1})\\
C_{5}: q_1=q_2=r_1+r_2-a_3+\eta=0 \hspace{-2cm} &\\
&\hspace{-3cm} U_5: (u_5,v_5,w_5, r_2)=(q_1, q_2q_1^{-1} ,(r_1+r_2-a_3+\eta)q_1^{-1},r_2)\\
C_{6}: Q_1=Q_2=r_1+r_2+\eta=0  &U_6: (u_6, v_6,w_6, r_2)=(Q_1, Q_2Q_1^{-1} ,(r_1+r_2+\eta)Q_1^{-1},r_2)\\
C_{7}: q_1-t=R_1=0  &U_7: (v_7, q_2,u_7, r_2)=((q_1-t)R_1^{-1}, q_2 ,R_1,r_2)\\
C_{8}: u_7=v_7-a_2 t=0  &U_8: (v_8, q_2,u_8, r_2)=((v_7-a_2 t)u_7^{-1}, q_2, u_7,r_2)\\
\\
C_{9}: q_2-1=R_2=0  &U_9: (q_1,v_9, r_1,u_9)=(q_1, (q_2-1)R_2^{-1} ,r_1,R_2)\\
C_{10}: u_9=v_9-a_4=0  &U_{10}: (q_1, v_{10},r_1 ,u_{10})=(q_1,(v_9-a_4)u_9^{-1}, r_1, u_9)\\
C_{11}: q_1-q_2=R_1=R_2=0  &\\
&\hspace{-3cm} U_{11}: (w_{11}, q_2, u_{11}, v_{11})=((q_1-q_2)R_1^{-1}, q_2 ,R_1,R_2R_1^{-1})\\
C_{12}: u_{11}=v_{11}+1=w_{11}-a_0q_2=0\hspace{-2cm} &  \\
&\hspace{-3cm} 
U_{12}: (w_{12}, q_2,u_{12}, v_{12})=((w_{11}-a_0q_2)u_{11}^{-1}, q_2 ,u_{11},(v_{11}+1)u_{11}^{-1}),
\end{array}\label{blowup_data}
\end{align}
where we write only one of new coordinate systems where the exceptional divisor is given by $u_i=0$ ($i=1,2,\cdots,12$). The other coordinate systems are automatically determined from the above data. For example, 
the other two coordinate systems for blowing up along $C_5$ are
$$U_5': (u_5',v_5',w_5', r_2)=(q_1q_2^{-1}, q_2 ,(r_1+r_2-a_3+\eta)q_2^{-1},r_2)$$
and
$$U_5'': (u_5'',v_5'',w_5'', r_2)=(q_1(r_1+r_2-a_3+\eta)^{-1}, q_2(r_1+r_2-a_3+\eta)^{-1} ,r_1+r_2-a_3+\eta,r_2),$$
where the exceptional divisor is given by $v_5'=0$ and $w_5''=0$ respectively.
More precisely, above coordinate systems are obtained only by blow-ups along open subset of $C_i$'s, but the other systems are also determined automatically by algebraic continuation (we assume $C_i$'s are irreducible).

\begin{theorem}
Let $\cX_{A}$ ($A=(a_0,a_1,\cdots, a_5,\eta, t)$) be a rational variety 
obtained by blowing-ups along $C_1,\cdots,C_{12}$ above. 
Each B\"aclund transformation $s_0,\cdots,s_5$ or $\rho$ is lifted to a pseudo-isomorphism from a rational variety $\cX_{A}$ to $\cX_{\bar{A}}$, where 
$\bar{A}=(\bar{a}_0,\bar{a}_1,\cdots, \bar{a}_5, \bar{\eta}, \bar{t})$ is given by  Table `Actions on parameters''.
\end{theorem}

\begin{proof}
It is confirmed by direct computation that the Jacobians of lifted mappings do not vanish.  
\end{proof}

\begin{remark}
Transformation $\pi$ is not lifted to a pseudo-isomorphism from $\cX_A$ to $\cX_{\bar{A}}$.
Indeed, $q_2=0$ (a hyper-surface) is mapped to 
$$(\bar{w}_{11}, \bar{q}_2, \bar{u}_{11},\bar{v}_{11})=
\left( \frac{r_1}{q_1}, \frac{1}{t},0,-1+\frac{t}{q_1}\right)
$$
with $\bar{q}_1-\bar{q}_2=\bar{R}_1=\bar{R}_2=0$.
Hence the image of $q_2=0$ is included in a two-dimentional subvariety
of $C_{11}$ such that the subvariety is different from $C_{12}$.
 Hence $q_2=0$ is contracted to lower dimensional variety in $\cX_{\bar{A}}$.
The same thing happens also for $Q_1=0$: which is contracted to 
a two-dimentional subvariety of $C_{11}$ different from $C_{12}$. 

Of course, there is possibility to be able to construct the space of initial conditions that
allow  $\pi$, but it needs several more blowing-ups, and the action of the root system would become more complicated.
\end{remark}

Let us denote $\cE_i$ as the class of total transform of the exceptional divisor obtained
by the blowing-up along $C_i$.
Then, since $C_5,C_6,C_{11},C_{12}$ are subvarieties of codimension 3, while
the others are of codimension 2, the anti-canonical divisor class is 
$$-K_{\cX_A}=2\cH_{q_1}+2\cH_{q_2}+2\cH_{r_1}+2\cH_{r_2}-
\sum_{i=1,2,3,4,7,8,9,10}\cE_i - 2 \sum_{i=5,6,11,12}\cE_i.$$ 

Moreover, since $C_8$, $C_{10}$ and $C_{12}$ are subvarieties of $\cE_7$, $\cE_9$ and $\cE_{11}$ respectively, $\cE_7-\cE_8$, $\cE_9-\cE_{10}$ and $\cE_{11}-\cE_{12}$ are effective divisor class in $\cX_A$. Hence, $-K_{\cX_A}$ is decomposed into irreducible divisors as 
\begin{align}
-K_{\cX_A}=&(\cH_{q_1}-\cE_1-\cE_5)+(\cH_{q_1}-\cE_2-\cE_6)
+(\cH_{q_2}-\cE_3-\cE_5)+(\cH_{q_2}-\cE_4-\cE_6)\nonumber \\
&+2(\cH_{r_1}-\cE_7-\cE_{11})+2(\cH_{r_2}-\cE_9-\cE_{11})\nonumber \\
&+ (\cE_7-\cE_8)+(\cE_9-\cE_{10})+2(\cE_{11}-\cE_{12}), \label{decomposition}
\end{align} 
where each irreducible divisor is explicitly written by coordinates as
\begin{align*}
\cH_{q_1}-\cE_1-\cE_5 :& q_1=0\\
\cH_{q_1}-\cE_2-\cE_6:& Q_1=0\\
\cH_{q_2}-\cE_3-\cE_5 :& q_2=0\\
\cH_{q_2}-\cE_4-\cE_6:& Q_2=0\\
\cH_{r_1}-\cE_7-\cE_{11}:& R_1=0\\
\cH_{r_2}-\cE_9-\cE_{11}:& R_2=0.
\end{align*}
Note that these subvarieties correspond to vertical leaves, i.e. the ordinary differential equations are not defined on these subvarieties.

\begin{remark}
Through the natural identification between exceptional divisors for different values of parameter $A$'s, we use the same symbol $\cE_i$'s for all $A$'s. 
\end{remark}

\section{The root system and the actions on the bilattice}
We can directly compute the actions $s_i$'s on the Picard group. However, in order to see geometric way of construction of the root system explicitly, let us reconstruct the B\"acklund transformations from a root system defined in the Neron-Severi bilattice. 

We use a higher dimensional analog of the notion of Cremona isometry introduced in \cite{CT2019}.
\begin{definition} An automorphism $s$ of the N\'eron-Severi bilattice is called a Cremona isometry if the following three properties are satisfied:\\
(a) $s$ preserves the intersection form;\\
(b) $s$ leaves the decomposition of $-K_{\cX}$ fixed;\\
(c) $s$ leaves the semigroup of effective classes of divisors invariant.
\end{definition}

Our aim is to realise the group of Cremona isometries as a root system,  
though it is the most heuristic part of this procedure.
Different to the two-dimensional case, we merely know the decomposition of the null-root (the anti-canonical divisor) in one of the dual spaces as \eqref{decomposition}, and hence, we can collect the vectors orthogonal to the elements of the decomposition only in the homology space.

Let us set the roots and co-roots as
\begin{align*}
\alpha_0=\cH_{q_1}+\cH_{q_2}-\cE_{5,6,11,12},\quad 
\alpha_1=\cH_{r_1}-\cE_{1,2},\quad
\alpha_2=\cH_{q_1}-\cE_{7,8},\\
\alpha_3=\cH_{r_1}+\cH_{r_2}-\cE_{5,6,11,12},\quad 
\alpha_4=\cH_{q_2}-\cE_{9,10},\quad
\alpha_5=\cH_{r_2}-\cE_{3,4}
\end{align*} 
and 
\begin{align*}
\check{\alpha}_0=h_{r_1}+h_{r_2}-e_{11,12},\quad 
\check{\alpha}_1=h_{q_1}-e_{1,2},\quad
\check{\alpha}_2=h_{r_1}-e_{7,8},\\
\check{\alpha}_3=h_{q_1}+h_{q_2}-e_{5,6},\quad 
\check{\alpha}_4=h_{r_2}-e_{9,10},\quad
\check{\alpha}_5=h_{q_2}-e_{3,4},
\end{align*}
where $\cE_{i_1,\dots,i_n}$ and $e_{i_1,\dots,i_n}$ are 
the abbreviations of $\cE_{i_1}+\cdots+\cE_{i_n}$ and $e_{i_1}+\cdots+e_{i_n}$
respectively.
Then, they constitute the root basis of type $A_5^{(1)}$ whose Cartan matrix and the Dynkin diagram are as follows.
\begin{center}
$\begin{bmatrix}
2&-1&0&0&0&-1\\
-1&2&-1&0&0&0\\
0&-1&2&-1&0&0\\
0&0&-1&2&-1&0\\
0&0&0&-1&2&-1\\
-1&0&0&0&-1&-2
\end{bmatrix}$
\raisebox{-17mm}{
\begin{tikzpicture}[line cap=round,line join=round,>=triangle 45,x=0.9cm,y=0.9cm]
\clip(-2,-2) rectangle (5.4,1.5);
\draw [line width=1.pt] (1.5,0.732)-- (-1.,-1.);
\draw [line width=1.pt] (-1.,-1.)-- (0.25,-1.);
\draw [line width=1.pt] (0.25,-1.)-- (1.5,-1.);
\draw [line width=1.pt] (1.5,-1.)-- (2.75,-1.);
\draw [line width=1.pt] (2.75,-1.)-- (4.0,-1.);
\draw [line width=1.pt] (1.5,0.732)-- (4.0,-1.);
\draw [fill=black] (1.5,0.732) circle (2.5pt);
\draw [fill=black] (-1.,-1.) circle (2.5pt);
\draw [fill=black] (0.25,-1.) circle (2.5pt);
\draw [fill=black] (1.5,-1) circle (2.5pt);
\draw [fill=black] (2.75,-1) circle (2.5pt);
\draw [fill=black] (4.0,-1) circle (2.5pt);
\begin{small}
\draw [color=black] (1.5,1.2) node {$\alpha_0$};
\draw [color=black] (-1.,-1.5) node {$\alpha_1$};
\draw [color=black] (0.25,-1.5) node {$\alpha_2$};
\draw [color=black] (1.5,-1.5) node {$\alpha_3$};
\draw [color=black] (2.75,-1.5) node {$\alpha_4$};
\draw [color=black] (4.,-1.5) node {$\alpha_5$};
\end{small}
\end{tikzpicture}}
\end{center} 

\begin{remark}
\begin{enumerate}
\item The decomposition \eqref{decomposition} constitutes of 9 irreducible divisors in rank 14 space (with null-vector). Hence, rank 6 of type $A_5^{(1)}$ is maximum.  
\item The number of free parameters of the space of initial conditions having the decomposition \eqref{decomposition} is 19: 2 for $C_i$ ($i=1,2,3,4,7,9,12$), and 1 for $C_i$ ($i=5,6,8,10,11$). We can reduce it using M\"obius transformation with respect to each coordinate to $19-3\times 4 =7$. We leave one of 12 variables (able to be fixed by M\"obius transformations) to be free for the continuous time variable $t$, while the sum of $a_i$'s is fixed to be 1.   
Hence we have 7 variables in total as $a_0, \cdots, a_5$, $\eta$ and $t$ in the formulation of Fuji-Suzuki \cite{FS2010}.
  
\end{enumerate}
\end{remark}

The actions of the roots on the N\'eron-Severi bilattice are given by the formulae
$$ s_i({\mathcal D})={\mathcal D} - 2\frac{\langle {\mathcal D},\check{\alpha}_i\rangle}{\langle \alpha_i, \check{\alpha}_i \rangle}\alpha_i,\quad 
s_i(d)=d - 2\frac{\langle \alpha_i, d \rangle}{\langle \alpha_i, \check{\alpha}_i \rangle }\check{\alpha}_i
$$  
for any ${\mathcal D} \in H^2(\cX_A,\Z)$ and $d \in H_2(\cX_A,\Z)$, while
$\rho$ acts on $H^2(\cX_A,\Z)$ as 
\begin{align}
&\begin{array}{lcl}
\cH_{q_1} \leftrightarrow \cH_{q_2},\quad \cH_{p_1} \leftrightarrow \cH_{p_2}
&\qquad&h_{q_1} \leftrightarrow h_{q_2},\quad h_{p_1} \leftrightarrow h_{p_2}\\
\cE_1\leftrightarrow \cE_3,\quad \cE_2\leftrightarrow \cE_4
&&e_1\leftrightarrow e_3,\quad e_2\leftrightarrow e_4\\
\cE_7\leftrightarrow \cE_9,\quad \cE_8\leftrightarrow \cE_{10}
&&e_7\leftrightarrow e_9,\quad e_8\leftrightarrow e_{10}
\end{array}
\end{align}

\noindent \textit{Translation}\ \\
A translation of $A_5^{(1)}$ 
\begin{align}
T_{\alpha_i}: \bar{\alpha}_i= \alpha_i+2\delta,
\quad \bar{\alpha}_{i\pm1}= \alpha_{i\pm 1}-\delta,\quad
\bar{\alpha}_j= \alpha_j \quad (|i-j|>1),
\end{align}
where $\delta=-K_{\cX}$, is realised as 
$$T_{\alpha_i}=s_{i+1}\circ s_{i+2}\circ s_{i+3}\circ s_{i+4}\circ s_{i+5}\circ s_{i+4}\circ s_{i+3}\circ s_{i+2}\circ s_{i+1}\circ s_i.$$
The action of $T_{\alpha_i}$ on the Picard group is given by Kac's formula
\begin{align}
T_{\alpha_i}^*({\mathcal D})={\mathcal D}+\langle {\mathcal D},\check{\delta}\rangle\alpha_i
+\left(\langle {\mathcal D},\check{\delta}\rangle+\frac{1}{2}
\langle \alpha_i, \check{\alpha}_i \rangle
\langle {\mathcal D},\check{\alpha}_i\rangle \right)\delta,
\end{align}
for any ${\mathcal D} \in H^2(\cX_A,\Z)$, where
\begin{align}
\check{\delta}=&\sum_{i=0}^5 \check{\alpha}_i=
2h_{q_1}+2h_{q_2}+2h_{r_1}+2h_{r_2}-\sum_{i=1}^{12}e_i.
\end{align}
Since the degrees of the iteration of this mapping with respect to
$q_1$, $q_2$, $r_1$, $r_2$
are given by the coefficients
of $\cH_{q_1}$,  $\cH_{q_2}$, $\cH_{r_1}$, $\cH_{r_1}$ in
\begin{align}
\left((T_{\alpha_i})^n\right)^*({\mathcal D})={\mathcal D}+n \langle {\mathcal D},\check{\delta}\rangle\alpha_i
+\left(n^2\langle {\mathcal D},\check{\delta}\rangle+\frac{n}{2}
\langle \alpha_i, \check{\alpha}_i \rangle
\langle {\mathcal D},\check{\alpha}_i\rangle \right)\delta
\end{align} 
with ${\mathcal D}=\cH_{q_1}, \cH_{q_2}, \cH_{p_1}, \cH_{p_2}$
(as explained in \S~1.3), they increase quadratically with respect to $n$.

\subsection*{Acknowledgments} The author thanks careful reading and thoughtful comments of the anonymous reviewer.


\end{document}